\documentclass[graybox, draft]{svmult}

\usepackage{mathptmx}       
\usepackage{helvet}         
\usepackage{courier}        
\usepackage{type1cm}        
\usepackage{makeidx}         
\usepackage{graphicx}        
\usepackage{multicol}        
\usepackage[bottom]{footmisc}

\makeindex             

\usepackage{xcolor}
\usepackage[obeyDraft, 
                bordercolor=TodoColour, 
                backgroundcolor=TodoColour,  
                linecolor=TodoColour, 
                textsize=todosize]{todonotes}

\definecolor{TodoColour}{HTML}{967B7D}
\definecolor{TodoColourC}{HTML}{FDB700}
\definecolor{EmphColour}{HTML}{9554B3} 
\definecolor{LinkColour}{HTML}{782075} 

\newcommand{\todoC}[2][\relax]{%
  \ifx#1\relax%
  \todo[color=TodoColourC, linecolor=TodoColourC, bordercolor=TodoColourC]{#2}%
  \else%
  \todo[color=TodoColourC, linecolor=TodoColourC, bordercolor=TodoColourC, #1]{#2}%
  \fi}

\usepackage{amsmath}
\usepackage{amssymb}

\newcommand{\autoref}[2][\relax]{%
  \textcolor{LinkColour}{\ifx#1\relax\relax\else#1~\fi\ref{#2}}}

\definecolor{LinkColour}{HTML}{2F4180}

\usepackage{tikz}
\usetikzlibrary{matrix}

\definecolor{MPurple}{HTML}{5A3A69}
\definecolor{MYellow}{HTML}{FCB500}
\definecolor{MCyan}{HTML}{41665A}
\definecolor{MYellowGrey}{HTML}{544F41}
\definecolor{MPurpleGrey}{HTML}{CAC6CC}

\newcommand{\Bigmid}{\Bigl.\;\Bigm|\;\Bigr.}
\newcommand{\Bignmid}{\Bigl.\;\not\hspace{0.2em}\Bigm|\;\Bigr.}

\newcommand{\A}{\mathbb{A}}
\newcommand{\EE}{\mathbb{E}}
\newcommand{\F}{\mathbb{F}}
\newcommand{\G}{\mathbb{G}}
\newcommand{\K}{\mathbb{K}}
\newcommand{\Q}{\mathbb{Q}}
\newcommand{\N}{\mathbb{N}}
\newcommand{\Z}{\mathbb{Z}}

\def\pisiSE{$\Pi\Sigma$}

\newcommand{\ie}{i.\,e.}
\newcommand{\wrt}{w.\,r.\,t.}
\newcommand{\eg}{e.\,g.}
\newcommand{\aka}{also known as}

\newcommand{\RHS}{right hand side}
\newcommand{\LHS}{left hand side}

\newcommand{\latin}[1]{\textit{#1}}

\newcommand{\qqtext}[1]{\qquad\text{#1}\qquad}

\newcommand{\qqforall}{\qquad\text{for all}\quad}

\DeclareMathOperator{\spread}{spread}
\DeclareMathOperator{\disp}{disp}
\DeclareMathOperator{\per}{per}
\DeclareMathOperator{\aperiodic}{ap}
\DeclareMathOperator{\LRCM}{LRCM}
\DeclareMathOperator{\const}{const}
\DeclareMathOperator{\diag}{diag}
\DeclareMathOperator{\lcm}{lcm}

\renewcommand{\th}{\raisebox{0.7ex}{\scriptsize th}}
\newcommand{\dtilde}[1]{\hat{#1}}

\newcommand{\PiExt}{$\Pi$-extension}
\newcommand{\SigmaExt}{$\Sigma$-extension}
\newcommand{\PSE}{$\Pi\Sigma$-extension}
\newcommand{\PSEUpcase}{$\Pi\Sigma$-Extension}
\newcommand{\PSEs}{\PSE{}s}
\newcommand{\PSEUpcases}{\PSEUpcase{}s}
\newcommand{\DF}{difference field}

\newcommand{\CV}[2]{{#1}^{#2}}

\newcommand{\Mat}[3]{{#1}^{#2\times#3}}
\newcommand{\MatGr}[2]{\mathrm{GL}_{#2}(#1)}

\providecommand{\ID}[1][\relax]{\mathbf{1}\ifx#1\relax\relax\else_{#1}\fi}
\def\auxZEROsubscript#1,#2;{_{#1\times#2}}
\providecommand{\ZERO}[1][\relax]{%
  \mathbf{0}\ifx#1\relax\relax\else\auxZEROsubscript#1;\fi}

\newcommand{\Ore}{\F(t)[\sigma]}
\newcommand{\OreInv}{\F(t)[\sigma^{-1}]}
\newcommand{\OLP}{\F(t)[\sigma,\sigma^{-1}]}

\begin{document}

\title*{Denominator Bounds for Systems of Recurrence Equations using
  \PSEUpcases\thanks{Supported by the Austrian Science Fund (FWF) grant SFB F50 (F5009-N15)}}
\titlerunning{Denominators of Recurrence Systems} 
\author{Johannes Middeke \and Carsten Schneider}
\institute{Johannes Middeke
  \at Research Institute for Symbolic Computation (RISC), 
  Johannes Kepler University, Altenbergerstra\ss{}e 69, A-4040 Linz, Austria
  \email{jmiddeke@risc.jku.at}
  \and
  Carsten Schneider
  \at Research Institute for Symbolic Computation (RISC), 
  Johannes Kepler University, Altenbergerstra\ss{}e 69, A-4040 Linz, Austria
  \email{cschneid@risc.jku.at} 
}
%
%
\maketitle

\abstract{We consider linear systems of recurrence equations whose coefficients are given 
in terms of indefinite nested sums and products covering, e.g., the harmonic numbers, hypergeometric products, $q$-hypergeometric products or their mixed versions. These linear systems are formulated in 
the setting of \PSEs\ and our goal is to find a denominator bound (\aka\ universal
  denominator) for the solutions; \ie, a non-zero polynomial $d$ such
  that the denominator of every solution of the system divides
  $d$. This is the first step in computing all rational solutions of
  such a rather general recurrence system. Once the denominator bound is known, the
  problem of solving for rational solutions is reduced to the problem
  of solving for polynomial solutions.}

 \abstract*{We consider linear systems of recurrence equations whose coefficients are given 
in terms of indefinite nested sums and products covering, e.g., the harmonic numbers, hypergeometric products, $q$-hypergeometric products or their mixed versions. These linear systems can be formulated in 
the setting of \PSEs\ and our goal is to find a denominator bound (\aka\ universal
  denominator) for the solutions; \ie, a non-zero polynomial $d$ such
  that the denominator of every solution of the system divides
  $d$. This is the first step in computing all rational solutions of
  such a rather general recurrence system. Once the denominator bound is known, the
  problem of solving for rational solutions is reduced to the problem
  of solving for polynomial solutions.}

\section{Introduction}

Difference equations are one of the central tools within symbolic summation. 
In one of its simplest forms, the telescoping equation plays a key role: given a sequence $f(k)$, find a solution $g(k)$ of
$$f(k)=g(k+1)-g(k).$$
Finding such a $g(k)$ in a given ring/field or in an appropriate extension of it (in which the needed sequences are represented accordingly) yields a closed form of the indefinite sum
$\sum_{k=a}^b f(k)=g(b+1)-g(a)$.
Slightly more generally, solving the creative telescoping and more generally the parameterized telescoping equation enable one to search for linear difference equations of definite sums. Finally, linear recurrence solvers enhance the the summation toolbox to find closed form solutions of definite sums. This interplay between the different algorithms to solve difference equations has been worked out in~\cite{Abramov:71,Gosper:78,Zeilberger:91,Petkov:92,AeqlB,Paule:95} for hypergeometric sums and has been improved, e.g., by difference ring algorithms~\cite{Karr81,Bron:00,Schneider:13a,Schneider:16a,Schneider:17} to the class of nested sums over hypergeometric products, $q$-hypergeometric products or their mixed versions, or by holonomic summation algorithms~\cite{Zeilberger:90a,Chyzak:00} to the class of sequences/functions that can be described by linear difference/differential equations.

More generally, coupled systems of difference equations are heavily used to describe problems coming from practical problem solving. E.g., big classes of Feynman integrals in the context of particle physics can be described by coupled systems of linear difference equations; for details and further references see~\cite{Schneider:16b}. Here one ends up at $n$ Feynman integrals $I_1(k),\dots,I_n(k)$ which are solutions of a coupled system. More precisely, we are given matrices $A_0(k),\dots A_l(k)\in\K(k)^{m\times n}$ with entries from the rational function field $\K(k)$, $\K$ a field containing the rational numbers, and a vector $b(k)$ of length $m$ in terms of nested sums over hypergeometric products such that the following coupled system holds:
\begin{equation}\label{Equ:FunctionForm}
A_l(k) \left(\begin{matrix} I_1(k+l)\\ \vdots\\ I_n(k+l)\end{matrix}\right)+A_{l-1}(k) \left(\begin{matrix} I_1(k+l-1)\\ \vdots\\ I_n(k+l-1)\end{matrix}\right)+\dots+A_0(k)\left(\begin{matrix} I_1(k)\\ \vdots\\ I_n(k)\end{matrix}\right)=b(k).
\end{equation}
Then one of the main challenges is to solve such a system, e.g., in terms of d'Alembertian~\cite{Abramov:94,Abramov:96} or Liouvillian solutions~\cite{Singer:99,Petkov:2013}. Furthermore, solving coupled systems arises as crucial subproblem within holonomic summation algorithms~\cite{Chyzak:00}. In many situations, one proceeds as follows to get the solutions of such a coupled system: first decouple the system using any of the algorithms described in~\cite{Barkatou:93,Zuercher:94,Abramov:96,myPhDthesis,BCP13} such that one obtains a scalar linear recurrence in only one of the unknown functions, say $I_1(k)$, and such that the remaining integrals $I_2(k),\dots I_n(k)$ can be expressed as a linear combination of the shifted versions of $I_1(k)$ and the entries of $b(k)$ over $\K(k)$. Thus solving the system~\eqref{Equ:FunctionForm} reduces to the problem to solve the derived linear recurrence and, if this is possible, to combine the solutions such that $I_1(k)$ can be expressed by them. Then given this solution, one obtains for free also the solutions of the remaining integrals $I_2(k),\dots,I_n(k)$.
This approach in general is often rather successful since one can rely on the very well explored solving algorithms~\cite{Abramov:89a,AbramovPaulePetkovsec,Abramov:94,Abramov:96,Abramov1995,Abramov:96,Hoeij:98,Singer:99,vanHoeij:99,Bron:00,Schneider:04b,Schneider:05a,Petkov:2013} to determine, e.g., d'Alembertian and Liouvillian solutions for scalar linear recurrence relations and can heavily use summation algorithms~\cite{Karr81,Schneider:07d,Schneider:15,Schneider:16a} to simplify the found solutions.

The main drawback of this rather general tactic of solving a decoupled system is efficiency. First, the decoupling algorithms themselves can be very costly; for further details see~\cite{BCP13}. Second, the obtained scalar recurrences have high orders with rather huge coefficients and the existing solving algorithms might utterly fail to find the desired solutions in reasonable time. Thus it is highly desirable to attack the original system~\eqref{Equ:FunctionForm} directly and to avoid any expensive reduction mechanisms and possible blow-ups to a big scalar equation. Restricting to the first-order case ($m=n=1$), this problem has been treated the first time in~\cite{AB98}. Given an invertible matrix $A(t)$ from $\K(t)^{n\times n}$, find all solutions $y(t)=(y_1(t),\dots,y_n(t))\in\K(t)^n$ such that
\begin{equation}\label{Equ:FirstOrderCase}
y(t+1)-A\,y(t)=0
\end{equation}
holds. As for many other symbolic summations approaches~\cite{Abramov:71,Gosper:78,Karr81,Abramov:89a,Petkov:92,Paule:95,AbramovPaulePetkovsec,Bron:00,Schneider:05a,Schneider:15} one follows the following strategy (sometimes the first step is hidden in certain normal-form constructions or certain reduction strategies):
\begin{enumerate}
\item Compute a universal denominator bound, i.e., a $d(t)\in\K[t]\setminus\{0\}$ such that for any solution $y(t)\in\K(t)^n$ of~\eqref{Equ:FirstOrderCase} we have $d(t)\,y(t)\in\K[t]^n$. 
\item Given such a $d(t)$, plugging $y(t)=\frac{y'(t)}{d(t)}$ into~\eqref{Equ:FirstOrderCase} yields an adapted system for the unknown polynomial $y'(t)\in\K[t]^n$. Now compute a degree bound, i.e., a $b\in\N$ such that the degrees of the entries in $y'$ are bounded by $b$.
\item Finally, inserting the potential solution $y'=y_0+y_1\,t+\dots+y_b\,t^b$ into the adapted system yields a linear system in the components $(y_{01}, \ldots, y_{0n},\dots,y_{b1},\ldots,y_{bn}) \in \K^{n(b+1)}$ of the unknown vectors $y_0,\ldots,y_b \in \K^n$. Solving this system yields all $y_0,\dots,y_b \in\K^{n}$ and thus all solutions $y(t)\in\K(t)^n$ for the original system~\eqref{Equ:FirstOrderCase}.
\end{enumerate}
For an improved version exploiting also ideas from~\cite{Hoeij:98} see~\cite{Abramov:11}. Similarly, the $q$-rational case (i.e., $t\mapsto q\,t$ instead of $t\mapsto t+1$) has been elaborated in~\cite{Abramov:99,Abramov2002}.
In addition, the higher~order case $m=n\in\N$ has been considered in~\cite{AbramovKhmelnov2012} for the rational case. 

In this article, we will push further the calculation of universal denominators (see reduction step (1)) to the general difference field setting of \pisiSE-fields~\cite{Karr81} and more general to the framework of \pisiSE-extensions~\cite{Karr81}. 
Here we will utilise similar as in~\cite{AbramovKhmelnov2012,AbramovBarkatou2014}
algorithms from~\cite{BeckermannChengLabahn2006} to transform in a preprocessing step the coupled system to an appropriate form. Given this modified system, we succeed in generalising compact formulas of universal denominator bounds from~\cite{Barkatou1999,CPS:08} to the general case of \pisiSE-fields. In particular, we generalise the available denominator bounds in the setting of \pisiSE-fields of~\cite{Bron:00,Schneider:04b} from scalar difference equations to coupled systems.
As consequence, the earlier work of the denominator bounding algorithms is covered in this general framework and big parts of the $q$-rational, multibasic and mixed multibasic case~\cite{BauerPetkovsek1999} for higher-order linear systems are elaborated. More generally, these denominator bound algorithms enable one to search for solutions of coupled systems~\eqref{Equ:FunctionForm} where the matrices $A_i(k)$ and the vector $b(k)$ might consist of expressions in terms of indefinite nested sums and products and the solutions might be derived in terms of such sums and products. Furthermore, these algorithms can be used to tackle matrices $A_i(k)$ which are not necessarily square. Solving such systems will play an important role for holonomic summation algorithms that work over such general difference fields~\cite{Schneider:05d}. In particular, the technologies described in the following can be seen as a first step towards new efficient solvers of coupled systems that arise frequently within the field of particle physics~\cite{Schneider:16b}.

The outline of the article is as follows. In Section~\ref{Sec:PiSigma} we will present some basic properties of \pisiSE-theory and will present our main result concerning the computation of the aperiodic part of a universal denominator of coupled systems in a \pisiSE-extension. In Section~\ref{Sec:OrePolynmials} we present some basic facts on Ore polynomials which we use as an algebraic model for recurrence operators and introduce some basic definitions for matrices. With this set up, we will show in Section~\ref{Sec:DenominatorBounds}
how the aperiodic part of a universal denominator can be calculated under the assumption that the coupled system is brought into particular regularised form. This regularisation is carried out in Section~\ref{sec:regularise} which relies on row reduction that will be introduced in Section~\ref{sec:row-red}. We present examples in Section~\ref{sec:examples} and conclude in Section~\ref{Sec:Conclusion} with a general method that enables one to search for solutions in the setting of \pisiSE-fields.

\section{Some \pisiSE-Theory and the Main Result}\label{Sec:PiSigma}

In the following we model the objects in~\eqref{Equ:FunctionForm}, i.e, in the entries of $A_0(k),\dots,A_l(k)$ and of $b(k)$ with elements from a field\footnote{Throughout this article, all fields contain the rational numbers $\Q$ as subfield.} $\F$. 
Further we describe the shift operation acting on these elements by a field automorphism $\sigma \colon \F \to \F$. In short, we call such a pair $(\F,\sigma)$ consisting of a field equipped with a field automorphism also a difference field.

\begin{example}\label{Exp:BaseField}
\begin{enumerate}
 \item Consider the rational function field $\F=\K(t)$ for some field $\K$ and the field automorphism $\sigma \colon \F \to \F$ defined by $\sigma(c)=c$ for all $c\in\K$ and $\sigma(t)=t+1$. $(\F,\sigma)$ is also called the rational difference field over $\K$.
 \item Consider the rational function field $\K=\K'(q)$ over the field $\K'$ and the rational function field $\F=\K(t)$ over $\K$. Further define the field automorphism $\sigma \colon \F \to \F$ defined by $\sigma(c)=c$ for all $c\in\K$ and $\sigma(t)=q\,t$. $(\F,\sigma)$ is also called the $q$-rational difference field over $\K$.
 \item Consider the rational function field $\K=\K'(q_1,\dots,q_e)$ over the field $\K'$ and the rational function field $\F=\K(t_1,\dots,t_e)$ over $\K$. Further define the field automorphism $\sigma \colon \F \to \F$ defined by $\sigma(c)=c$ for all $c\in\K$ and $\sigma(t_i)=q_i\,t_i$ for all $1\leq i\leq e$. $(\F,\sigma)$ is also called the $(q_1,\dots,q_e)$-multibasic rational difference field over $\K$.
 \item Consider the rational function field $\K=\K'(q_1,\dots,q_e)$ over the field $\K'$ and the rational function field $\F=\K(t_1,\dots,t_e,t)$ over $\K$. Further define the field automorphism $\sigma \colon \F \to \F$ defined by $\sigma(c)=c$ for all $c\in\K$, $\sigma(t)=t+1$ and $\sigma(t_i)=q_i\,t_i$ for all $1\leq i\leq e$. $(\F,\sigma)$ is also called the mixed $(q_1,\dots,q_e)$-multibasic rational difference field over $\K$.
\end{enumerate}
\end{example}

More generally, we consider difference fields that are built by the following type of extensions. Let $(\F,\sigma)$ be a \DF; \ie, a field $\F$ together with an
automorphism $\sigma \colon \F \to \F$. Elements of $\F$ which are
left fixed by $\sigma$ are referred to as \emph{constants}. We denote
the set of all constants by 
$$\const \F=\{c\in\F\mid \sigma(c)=c\}.$$ 
A \emph{\PSE} $(\F(t),\sigma)$ of $(\F,\sigma)$ is
given by the rational function field $\F(t)$ in the indeterminate $t$
over $\F$ and an extension of $\sigma$ to $\F(t)$ which can be built as follows:
either
\begin{enumerate}
\item $\sigma(t) = t + \beta$ for some $\beta \in \F \setminus \{0\}$
  (a $\Sigma$-monomial) or
\item $\sigma(t) = \alpha t$ for some $\alpha \in \F \setminus \{0\}$
  (a $\Pi$-monomial)
\end{enumerate}
where in both cases we require that $\const \F(t) = \const \F$. More generally, we consider a tower $(\F(t_1)\dots(t_e),\sigma)$ of such extensions where the $t_i$ are either $\Pi$-monomials or $\Sigma$-monomials adjoined to the field $\F(t_1)\dots(t_{i-1})$ below. Such a construction is also called a \pisiSE-extension $(\F(t_1)\dots(t_e),\sigma)$ of $(\F,\sigma)$. If $\F(t_1)\dots(t_e)$ consists only of $\Pi$-monomials or of $\Sigma$-monomials, it is also called a $\Pi$- or a $\Sigma$-extension. If $\F=\const \F$, $(\F(t_1)\dots(t_e),\sigma)$ is called a \pisiSE-field over $\F$.

Note that all difference fields from Example~\ref{Exp:BaseField} are \pisiSE-fields over $\K$. Further note that \pisiSE-extensions enable one to model indefinite nested sums and products that may arise as rational expressions in the numerator and denominator. See \cite{Karr81} or \cite{Schneider:13a} for examples of how that modelling works.

\medskip

Let $(\F,\sigma)$ be an arbitrary difference field and $(\F(t),\sigma)$ be a \pisiSE-extension of $(\F,\sigma)$. In this work, we take a look at systems of the form
\begin{equation}
  \label{eq:system}
  A_\ell \sigma^\ell y + \ldots + A_1 \sigma y + A_0 y = b
\end{equation}
where $A_0,\ldots,A_\ell \in \Mat{\F[t]}mn$ and $b \in
\CV{\F[t]}m$. Our long-term goal is to find all rational solutions for such
a system, \ie, rational vectors $y \in \CV{\F(t)}n$ which
satisfy~\eqref{eq:system} following the three steps presented in the introduction.
In this article we will look at the first step: compute a so-called
\emph{denominator bound} (\aka\ a \emph{universal denominator}). This
is a polynomial $d \in \F[t]\setminus\{0\}$ such that $d y \in \CV{\F[t]}n$ is
polynomial for all solutions $y$ of~\eqref{eq:system}. Once that is
done, we can simply substitute the denominator bound into the system
and then it will be sufficient to search for polynomial solutions. 
In future work, it will be a key challenge to derive such degree bounds; compare the existing results~\cite{Karr81,Schneider:01,Schneider:05b} for scalar equations. Degree bounds for the rational case ($l=1$) and the $q$-rational case ($l$ arbitrarily) applied to the system~\eqref{eq:system} can be found in~\cite{AbramovBarkatou1998} and~\cite{Middeke:17}, respectively. Once a
degree bound for the polynomial solutions is known, 
the latter problem translates to solving linear systems over $\F$ if $\F=\const\F$. Otherwise, one can apply similar strategies as worked out in~\cite{Karr81,Bron:00,Schneider:05a} to reduce the problem to find polynomial solutions to the problem to solve coupled systems in the smaller field $\F$.  Further comments on this proposed machinery will be given in the conclusion.

In order to derive our denominator bounds for system~\eqref{eq:system}, we rely heavily on the following concept~\cite{Abramov:71,Bron:00}. Let $a,b \in \F[t]\setminus\{0\}$ be two non-zero polynomials. We
define the \emph{spread} of $a$ and $b$ as
\begin{equation*}
  \spread(a,b) = \{ k \geq 0 \mid \gcd(a, \sigma^k(b)) \notin \F \}.
\end{equation*}
In this regard note that $\sigma^k(b)\in\F[t]$ for any $k\in\Z$ and $b\in\F[t]$. In particular, if $b$ is an irreducible polynomial, then also $\sigma^k(b)$ is an irreducible polynomial.\\
The \emph{dispersion} of $a$ and $b$ is defined as the maximum of the
spread, \ie, we declare
\begin{math}
  \disp(a,b) = \max \spread(a,b)
\end{math}
where we use the conventions $\max \emptyset = -\infty$ and $\max S =
\infty$ if $S$ is infinite. As an abbreviation we will sometimes use
$\spread(a) = \spread(a,a)$ and similarly $\disp(a) = \disp(a,a)$.
We call $a\in\F[t]$ periodic if $\disp(a)$ is infinite and aperiodic if $\disp(a)$ is finite.

It is shown in~\cite{Karr81,Bron:00} (see also~\cite[Theorem~2.5.5]{Schneider:01}) that in the case of
\SigmaExt{}s the spread of two polynomials will always be a finite set
(possibly empty). For \PiExt{}s the spread will certainly be infinite
if $t \mid a$ and $t \mid b$ as $\sigma^k(t) \mid t$ for all $k$. It
can be shown in~\cite{Karr81,Bron:00} (see also~\cite[Theorem~2.5.5]{Schneider:01}), however, that this is the only problematic
case. Summarising, the following property holds.

\begin{lemma}\label{Lemma:PeriodicChar}
Let $(\F(t),\sigma)$ be a \pisiSE-extension of $(\F,\sigma)$ and $a\in\F[t]\setminus\{0\}$. Then $a$ is periodic if and only if $t$ is a $\Pi$-monomial and $t\mid a$.
\end{lemma}

\noindent This motives the following definition.

\begin{definition}
Let $(\F(t),\sigma)$ be a \pisiSE-extension of $(\F,\sigma)$ and $a\in\F[t]\setminus\{0\}$. We define the periodic part of $a$ as 
$$\per(a)=\begin{cases} 1&\text{ if $t$ is a $\Sigma$-monomial},\\
               t^m& \text{ if $t$ is a $\Pi$-monomial and $m\in\N$ is maximal s.t.\ $t^m\mid a$}
              \end{cases}.$$
and the aperiodic part as $\aperiodic(a)=\frac{a}{per(a)}$.
\end{definition} 
Note that $\aperiodic(a)=a$ if $t$ is a $\Sigma$-monomial.  In this
article we will focus on the problem to compute the aperiodic part of a
denominator bound $d$ of the system~\eqref{eq:system}. Before we state
our main result, we will have to clarify what me mean by the
denominator of a vector.

\begin{definition}\label{def:reduced}
  Let $y \in \CV{\F(t_1,\ldots,t_e)}n$ be a rational column vector. We
  say that $y = d^{-1} z$ is a \emph{reduced representation} for $y$
  if $d \in \F[t_1,\ldots,t_e]\setminus\{0\}$ and $z=(z_1,\dots,z_n)\in \CV{\F[t_1,\ldots,t_e]}n$
  are such that\footnote{If $z$ is the zero vector, then the assumption $\gcd(z,d)=1$ implies $d=1$.} $\gcd(z,d) = \gcd(z_1,\ldots,z_n,d) = 1$.
\end{definition}

With all the necessary definitions in place, we are ready to state the
main result. Its proof will take up the remainder of this paper.

\begin{theorem}\label{Thm:GlobalMain}
  Let $(\F(t),\sigma)$ be a \PSE\ of $(\F,\sigma)$ and let
  $A_0,\dots,A_l\in\F[t]^{m\times n}$, $b\in\F[t]^m$. If one can
  compute the dispersion of polynomials in $\F[t]$, then one can
  compute the aperiodic part of a denominator bound
  of~\eqref{eq:system}. This means that one can compute a
  $d\in\F[t]\setminus\{0\}$ such that for any solution $q^{-1} p
  \in\CV{\F(t)}n$ of~\eqref{eq:system} with $q^{-1} p$ being in
  reduced representation we have $\aperiodic(q)\mid d$.
\end{theorem}

Note that such a $d$ in Theorem~\ref{Thm:GlobalMain}  forms a complete denominator bound if $t$ is a $\Sigma$-monomial. Otherwise, if $t$ is a $\Pi$-monomial, there exists an $m\in\N$ such that $t^m\,d$ is a denominator bound. Finding such an $m$ algorithmically in the general \pisiSE-case is so far an open problem. For the $q$-rational case we refer to~\cite{Middeke:17}.

In order to prove Theorem~\ref{Thm:GlobalMain}, we will perform a preprocessing step and regularise the system~\eqref{eq:system} to a more suitable form (see Theorem~\ref{thm:regularise} in Section~\ref{sec:regularise}); for similar strategies to accomplish such a regularisation see~\cite{AbramovKhmelnov2012,AbramovBarkatou2014}. Afterwards, we will apply Theorem~\ref{thm:main} in Section~\ref{Sec:DenominatorBounds} which is a generalisation of~\cite{Barkatou1999,CPS:08}. Namely, besides computing the  dispersion in $\F[t]$ one only has to compute certain $\sigma$- and $\gcd$-computations in $\F[t]$ in order to derive the desired aperiodic part of the universal denominator bound. 

Summarising, our proposed denominator bound method is applicable if the dispersion can be computed. To this end, we will elaborate under which assumptions the dispersion can be computed in $\F[t]$. Define for $f\in\F\setminus\{0\}$ and $k\in\Z$ the following functions:
\small
\begin{equation*}
f_{(k,\sigma)}:=\begin{cases}
f\sigma(f)\dots\sigma^{k-1}(f)&\text{if }k>0\\
1&\text{if }k=0\\
\frac{1}{\sigma^{-1}(f)\dots\sigma^{-k}(f)}&\text{if }k<0,
\end{cases}\quad
f_{\{k,\sigma\}}:=\begin{cases}
f_{(0,\sigma)}+f_{(1,\sigma)}+\dots+f_{(k-1,\sigma)}&\text{if }k>0\\
0&\text{if }k=0\\
-(f_{(-1,\sigma)}+\dots+f_{(k,\sigma)})&\text{if }k<0.
\end{cases}
\end{equation*}
\normalsize
Then analysing Karr's algorithm~\cite{Karr81} one can extract the following (algorithmic) properties that are relevant to calculate the dispersion in \pisiSE-extensions; compare~\cite{KS:06}. 

\begin{definition}\label{Def:SigmaCompDF}
$(\F,\sigma)$ is weakly $\sigma^*$-computable if the following holds.

\vspace*{-0.2cm}

\begin{enumerate}
\item There is an algorithm that factors multivariate polynomials
over $\F$ and that solves linear systems with multivariate rational functions over $\F$.

\item $(\F,\sigma^r)$ is torsion free for all $r\in\Z$, 
i.e., for all $r\in\Z$, for all
$k\in\Z\setminus\{0\}$ and all $g\in\F\setminus\{0\}$ the equality
$\big(\frac{\sigma^r(g)}{g}\big)^k=1$ implies $\frac{\sigma^r(g)}{g}=1$.

\item {\it $\Pi$-Regularity.} Given $f,g\in\F$ with $f$ not a root of
unity, there is at most one $n\in\Z$ such that $f_{(n,\sigma)}=g$.
There is an algorithm that finds, if possible, this $n$.

\item {\it $\Sigma$-Regularity.} Given $k\in\Z\setminus\{0\}$ and $f,g\in\F$ with $f=1$ or $f$ not a root of unity,
there is at most one $n\in\Z$ such that $f_{\{n,\sigma^k\}}=g$.
There is an algorithm that finds, if possible, this $n$.
\end{enumerate}
\end{definition}

\noindent Namely, we get the following result based on Karr's reduction algorithms.

\begin{lemma}\label{Lemma:LiftComputable}
Let $(\F(t),\sigma)$ be a \pisiSE-extension of $(\F,\sigma)$. Then the following holds.
\begin{enumerate}
 \item If $(\F,\sigma)$ is weakly $\sigma^*$-computable, one can compute the spread and dispersion of two polynomials $a,b\in\F[t]\setminus\F$.
 \item If $(\F,\sigma)$ is weakly $\sigma^*$-computable, $(\F(t),\sigma)$ is weakly $\sigma^*$-computable.
\end{enumerate}
\end{lemma}
\begin{proof}
(1) By Lemma~1 of~\cite{Schneider:04b} the spread is computable if the shift equivalence problem is solvable. This is possible if $(\F,\sigma)$ is weakly $\sigma^*$-computable; see Corollary~1 of~\cite{KS:06} (using heavily results of~\cite{Karr81}).\\
(2) holds by Theorem~1 of~\cite{KS:06}.\qed
\end{proof}

\noindent Thus by the iterative application of Lemma~\ref{Lemma:LiftComputable} we end up at the following result that supplements our Theorem~\ref{Thm:GlobalMain}.

\begin{corollary}\label{Cor:DispersionInExtension}
Let $(\F(t),\sigma)$ be a \pisiSE-extension of $(\F,\sigma)$ where $(\F,\sigma)$ with $\F=\G(t_1)\dots(t_e)$ is a \pisiSE-extension of a weakly $\sigma^*$-computable difference field $(\G,\sigma)$. Then the dispersion of two polynomials $a,b\in\F[t]\setminus\F$ is computable.
\end{corollary}

\noindent Finally, we list some difference fields $(\G,\sigma)$ that one may choose for Corollary~\ref{Cor:DispersionInExtension}. Namely, the following ground fields $(\G,\sigma)$ are weakly $\sigma^*$-computable.
\begin{enumerate}
 \item By~\cite{Schneider:05c} we may choose $\const{\G}=\G$ where $\G$ is a rational function field over an algebraic number field; note that $(\F,\sigma)$ is a \pisiSE-field over $\G$.
 \item By~\cite{KS:06} $(\G,\sigma)$ can be a free difference field over a constant field that is weakly $\sigma^*$-computable (see item 1).
 \item By~\cite{Schneider:07f} $(\G,\sigma)$ can be radical difference field over a constant field that is weakly $\sigma$-computable (see item 1).
 \end{enumerate}

\noindent Note that all the difference fields introduced in Example~\ref{Exp:BaseField} are \pisiSE-fields which are weakly $\sigma^*$-computable if the constant field $\K$ is a rational function field over an algebraic number field (see item~1 in the previous paragraph) and thus the dispersion can be computed in such fields. For the difference fields given in Example~\ref{Exp:BaseField} one may also use the optimised algorithms worked out in~\cite{BauerPetkovsek1999}. 

Using Theorem~\ref{Thm:GlobalMain} we obtain immediately the following multivariate case in the setting of \pisiSE-extensions which can be applied for instance for the multibasic and mixed multibasic rational difference fields defined in Example~\ref{Exp:BaseField}.

\begin{corollary}\label{Cor:PiSigmaNested}
Let $(\EE,\sigma)$ be a \pisiSE-extension of $(\F,\sigma)$ with $\EE=\F(t_1)(t_2)\dots(t_e)$
where $\sigma(t_i)=\alpha_i\,t_i+\beta_i$ ($\alpha_i\in\F^*$, $\beta_i\in\F$) for $1\leq i\leq e$. Let $A_0,\dots,A_l\in\EE^{m\times n}$, $b\in\EE$. Then there is an algorithm that computes a $d\in\F[t_1,\dots,t_e,t]\setminus\{0\}$ such that $d':=t_1^{m_1}\dots t_e^{m_e}\,d$ is a universal denominator bound of system~\eqref{eq:system}
for some $m_1,\dots,m_e\in\N$ where $m_i=0$ if $t_i$ is a $\Sigma$-monomial. That is, for any solution $y=q^{-1} p \in\CV{\F}n$ of~\eqref{eq:system} in reduced representation we have that $q\mid d'$.
\end{corollary}
\begin{proof}
Note that one can reorder the generators in $\EE=\F(t_1,\dots,t_e)$ without changing the constant field $\const\EE=\const\F$. Hence for any $i$ with $1\leq i\leq e$,
$(\A_i(t_i),\sigma)$ is a \pisiSE-extension of $(\A_i,\sigma)$ with $\A_i=\F(t_1)\dots(t_{i-1})(t_{i+1})\dots(t_e)$.
Thus for each $i$ with $1\leq i\leq e$,
we can apply Theorem~\ref{Thm:GlobalMain} (more precisely, Theorems~\ref{thm:regularise} and~\ref{thm:main} below) to compute the aperiodic part $d_i\in\A_i[t_i]\setminus\{0\}$ of a denominator bound of~\eqref{eq:system}. 
W.l.o.g.\ we may suppose that $d_1,\dots,d_e\in\A:=\F[t_1,\dots,t_e]$; otherwise, one clears denominators: for $d_i$ one uses a factor of $\A_i$). Finally, compute $d:=\lcm(d_1,\dots,d_e)$ in $\A$. Suppose that $d\,t_1^{m_1}\dots t_e^{m_e}$ is not a denominator bound for any choice $m_1,\dots,m_e\in\N$ where for $1\leq i\leq e$, $m_i=0$ if $t_i$ is a $\Sigma$-monomial. 
Then we find a solution $y=q^{-1} p$ of~\eqref{eq:system}
in reduced representation with $p \in \CV{\A}n$ and $q\in\A$ and an irreducible
$h\in\A\setminus\F$ with $h\mid q$ and $h\nmid d$ where
$h\neq t_i$ for all $i$ where $t_i$ is a $\Pi$-monomial. Let $j$ with $1\leq j\leq e$ such that $h\in\A_j[t_j]\setminus\A_j$. Since $d_j$ is the aperiodic part of a denominator bound \wrt\ $t_j$, 
and the case $h=t_j$ is excluded if $t_j$ is a $\Pi$-monomial, it follows that $h\,w=d_j$ for some $w\in\A_j[t_j]$.
Write $w=\frac{u}{v}$ with $u\in\A$ and $v\in\A_j$.
Since $d_j\in\A$, $h\,w\in\A$ and thus the factor $v\in\A$ must be contained in $h\in\A$. But since $h$ is irreducible in $\A$, $v\in\F\setminus\{0\}$ and thus $w\in\A$. Hence $h$ divides $d_j$ and thus it divides also $d=\lcm(d_1,\dots,d_e)$ in $\A$, a contradiction.\qed
\end{proof}

\section{Operators, Ore Polynomials, and Matrices}\label{Sec:OrePolynmials}

For this section, let $(\F, \sigma)$ be a fixed difference field. An
alternative way of expressing the system~\eqref{eq:system} is to use
operator notation. Formally, we consider the ring of \emph{Ore
  polynomials} $\Ore$ over the rational functions $\F(t)$ \wrt\ the
automorphism $\sigma$ and the trivial
$\sigma$-derivation\footnote{Some authors would denote $\Ore$ by the
  more precise $\F(t)[\sigma;\sigma,0]$.}. Ore polynomials are named
after {\O}ystein Ore who first described them in his paper
\cite{Ore33}. They provide a natural algebraic model for linear
differential, difference, recurrence of $q$-difference operators (see,
\eg, \cite{Ore33}, \cite{BP96}, \cite{ChyzakSalvy98},
\cite{AbramovLeLi2005} and the references therein).

\medskip

We briefly recall the definition of Ore polynomials and refer to the
aforementioned papers for details: As a set they consist of all
polynomial expressions
\begin{equation*}
  a_\nu \sigma^\nu + \ldots + a_1 \sigma + a_0
\end{equation*}
with coefficients in $\F(t)$ where we regard $\sigma$ as a
variable\footnote{A more rigorous way would be to introduce a new
  symbol for the variable. However, a lot of authors simply use the
  same symbol and we decided to join them.}. Addition of Ore
polynomials works just as for regular polynomials. Multiplication on
the other hand is governed by the \emph{commutation rule}
\begin{equation*}
  \sigma\cdot a = \sigma(a) \cdot \sigma
  \qqforall
  a \in \F(t).
\end{equation*}
(Note that in the above equation $\sigma$ appears in two different
roles: as the Ore variable and as automorphism applied to $a$.)  Using
the associative and distributive law, this rule lets us compute
products of arbitrary Ore polynomials. It is possible to show that
this multiplication is well-defined and that $\Ore$ is a
(non-commutative) ring (with unity).

For an operator $L = a_\nu \sigma^\nu + \ldots + a_0 \in
\Ore$ we declare the \emph{application} of $L$ to a rational
function $\alpha \in \F(t)$ to be
\begin{equation*}
  L(\alpha) 
  = a_\nu \sigma^\nu(\alpha) + \ldots + a_1 \sigma(\alpha) + a_0 \alpha.
\end{equation*}
Note that this turns $\F(t)$ into a left $\Ore$-module. We extend this
to matrices of operators by setting
\begin{math}
  L(\alpha) = \bigl(\sum_{j=1}^n L_{ij}(\alpha_j)\bigr)_j
\end{math}
for a matrix $L = (L_{ij})_{ij} \in \Mat{\Ore}mn$ and a
vector of rational functions $\alpha = (\alpha_j)_j \in
\CV{\Ore}n$. It is easy to see that the action of
$\Mat{\Ore}mn$ on $\CV{\F(t)}n$ is linear over $\const \F$.
With this notation, we can express the system~\eqref{eq:system} simply
as $A(y) = b$ where
\begin{equation*}
  A = A_\ell \sigma^\ell + \ldots + A_1 \sigma + A_0 
  \in \Mat{\Ore}mn.
\end{equation*}

The powers of $\sigma$ form a (left and right) Ore set within $\Ore$
(see, \eg, \cite[Chapter~5]{CohnRT} for a definition and a brief
description of localisation over non-commutative rings). Thus, we may
localise by $\sigma$ obtaining the \emph{Ore Laurent polynomials}
$\OLP$. We can extend the action of $\Ore$ on $\F(t)$ to $\OLP$ in the
obvious way.

\medskip

We need to introduce some notation and naming conventions. We denote
the $n$-by-$n$ identity matrix by $\ID[n]$ (or simply $\ID$ if the
size is clear from the context). Similarly $\ZERO[m,n]$ (or just
$\ZERO$) denotes the $m$-by-$n$ zero matrix. With
$\diag(a_1,\ldots,a_n)$ we mean a diagonal $n$-by-$n$ matrix with the
entries of the main diagonal being $a_1,\ldots,a_n$.

We say that a matrix or a vector is \emph{polynomial} if all its
entries are polynomials in $\F[t]$; we call it \emph{rational} if its
entries are fractions of polynomials; and we speak of \emph{operator}
matrices if its entries are Ore or Ore Laurent polynomials.

Let $M$ be a square matrix over $\F[t]$ (or $\Ore$ or $\OLP$). We say
that $M$ is \emph{unimodular} if $M$ possesses a (two-sided) inverse
over $\F[t]$ (or $\Ore$ or $\OLP$, respectively). We call $M$
\emph{regular}, if its rows are linearly independent over $\F[t]$ (or
$\Ore$ or $\OLP$, respectively) and \emph{singular} if they are not
linearly independent. For the special case of a polynomial matrix $M
\in \Mat{\F[t]}nn$, we can characterise these concepts using
determinants\footnote{The other two rings do not admit determinants
  since they lack commutativity.}: here, $M$ is singular if $\det M =
0$; regular if $\det M \neq 0$; and unimodular if $\det M \in
\F\setminus\{0\}$. Another equivalent characterisation of regular
polynomial matrices is that they have a rational inverse $M^{-1} \in
\Mat{\F(t)}nn$.

We denote the set of all unimodular polynomial matrices by
$\MatGr{\F[t]}n$ and that of all unimodular operator matrices by
$\MatGr{\Ore}n$ or by $\MatGr{\OLP}n$. We do not have a special
notation for the set of regular matrices.

\begin{remark}\label{rem:ore}
  Let $A \in \Mat{\Ore}mn$ and $b \in \CV{\F(t)}m$. Assume that we are
  given two unimodular operator matrices $P \in \MatGr{\OLP}m$ and $Q
  \in \MatGr{\OLP}n$. Then the system $A(y) = b$ has the solution $y$
  if and only if $(P A Q)(\tilde y) = P(b)$ has the solution $\tilde y
  = Q^{-1}(y)$: Assume first that $A(y) = b$. Then also $P(A(y)) = (P
  A)(y) = P(b)$ and furthermore we have $P(b) = (P A)(y) = (P A)(Q
  Q^{-1}(y)) = (P A Q)(Q^{-1}(y)) = (P A Q)(\tilde y)$. Because $P$
  and $Q$ are unimodular, we can easily go back as well. Thus, we can
  freely switch from one system to the other.
\end{remark}

\begin{definition}\label{def:related}
  We say that the systems $A(y) = b$ and $(P A Q)(\tilde y) = P(b)$ in
  \autoref[Remark]{rem:ore} are \emph{related} to each other.
\end{definition}

\section{Denominator Bounds for Regularised Systems}\label{Sec:DenominatorBounds}

Let $(\F(t), \sigma)$ be again a \pisiSE-extension of $(\F,\sigma)$. Recall that we we
are considering the system~\eqref{eq:system} which has the form
\begin{math}
  A_\ell \sigma^\ell y + \ldots + A_1 \sigma y + A_0 y = b
\end{math}
where $A_0,\ldots,A_\ell \in \Mat{\F[t]}mn$ and $b \in
\CV{\F[t]}m$. We start this section by identifying systems with good
properties.

\begin{definition}\label{def:regular}
  We say that the system in equation~\eqref{eq:system} is
  \emph{head regular} if $m = n$ (\ie, all the matrices are square)
  and $\det A_\ell \neq 0$.
\end{definition}

\begin{definition}\label{def:suitable}
  We say that the system in equation~\eqref{eq:system} is
  \emph{tail regular} if $m = n$ and $\det A_0 \neq 0$.
\end{definition}

\begin{definition}\label{def:full.reg}
  The system $A(y) = b$ in~equation~\eqref{eq:system} is called
  \emph{fully regular} if it is head regular and there exists a
  unimodular operator matrix $P \in \MatGr{\OLP}n$ such that the
  related system $(P A)(\tilde y) = P(b)$ is tail regular.
\end{definition}

We will show later in \autoref[Section]{sec:regularise} that any head
regular system is actually already fully regular and how the
transformation matrix $P$ from \autoref[Definition]{def:full.reg} can
be computed.

Moreover, in \autoref[Definition]{def:full.reg}, we can always choose
$P$ in such a way that the coefficient matrices $\tilde{A}_0, \ldots,
\tilde{A}_{\tilde\ell}$ and the right hand side of the related system
$(P A)(\tilde y) = P(b)$ are polynomial: simply multiplying by a
common denominator will not change the unimodularity of $P$.

The preceding \autoref[Definition]{def:full.reg} is very similar to
\emph{strongly row-reduced} matrices
\cite[Def.~4]{AbramovBarkatou2014}. The main difference is that we
allow an arbitrary transformation $P$ which translates between a head
and tail regular system while \cite{AbramovBarkatou2014} require their
transformation (also called $P$) to be of the shape
$\diag(\sigma^{m_1},\ldots,\sigma^{m_n})$ for some specific exponents
$m_1,\ldots,m_n \in \Z$. At this time, we do not know which of the two
forms is more advantageous; it would be an interesting topic for
future research to explore whether the added flexibility that our
definition gives can be used to make the algorithm more efficient.

\begin{remark}\label{rem:full.reg}
  In the situation of \autoref[Definition]{def:full.reg}, the
  denominators of the solutions of the original system $A(y) = b$ and
  the related system $\tilde{A}(\tilde{y}) = \tilde{b}$ are the same:
  By \autoref[Remark]{rem:ore}, we know that $y$ solves the original
  system if and only if $\tilde{y}$ solves the related system. The
  matrix $Q$ of \autoref[Remark]{rem:ore} is just the identity in this
  case.
\end{remark}

\smallskip

We are now ready to state the main result of this section. For the rational difference field this result appears in various specialised forms. E.g., the version $m=n=1$ can be also found in~\cite{CPS:08} and gives an alternative description of Abramov's denominator bound for scalar recurrences~\cite{Abramov:89a}. Furthermore, the first order case $l=1$ can be rediscovered also in~\cite{Barkatou1999}.

\begin{theorem}\label{thm:main}
  Let the system in~equation~\eqref{eq:system} be fully regular,
  and let $y = d^{-1} z \in \CV{\F(t)}n$ be a solution in reduced
  form. Let $(P A)(\tilde y) = P(b)$ be a tail regular related system
  with trailing coefficient matrix $\tilde{A}_0 \in
  \Mat{\F[t]}nn$. Let $m$ be the common denominator of $A_\ell^{-1}$
  and let $p$ be the common denominator of $\tilde{A}_0^{-1}$. Then
  \begin{equation}\label{Equ:FormualForD}
    \disp(\aperiodic(d)) 
    \leq \disp(\sigma^{-\ell}(\aperiodic(m)), \aperiodic(p)) = D
  \end{equation}
  and
  \begin{equation}\label{Equ:DenBoundFormula}
    \aperiodic(d) 
    \Bigmid
    \gcd\Bigl(
    \prod_{j=0}^D \sigma^{-\ell-j}(\aperiodic(m)),
    \prod_{j=0}^D \sigma^j(\aperiodic(p))
    \Bigr).
  \end{equation}
\end{theorem}

We will show in Section~\ref{sec:regularise} that any coupled system of the form~\eqref{eq:system} can be brought to a system which is fully regular and which contains the same solutions as the original system. Note further that the denominator bound of the aperiodic part given on the right hand side of~\eqref{Equ:DenBoundFormula} can be computed if the dispersion of polynomials in $\F[t]$ (more precisely, if $D$) can be computed.
Summarising, Theorem~\ref{Thm:GlobalMain} is established if Theorem~\ref{thm:main} is proven and if the transformation of system~\eqref{Equ:DenBoundFormula} to a fully regular version is worked out in Section~\ref{sec:regularise}.

\begin{remark}\label{Remark:ImprovedMultivariateDen}
Let $(\F(t_1)\dots(t_e),\sigma)$ be a \pisiSE-extension of $(\F,\sigma)$ with $\sigma(t_i)=\alpha_i\,t_i+\beta_i$ ($\alpha_i\in\F^*$, $\beta_i\in\F$) for $1\leq i\leq e$. In this setting a multivariate aperiodic denominator bound $d\in\F[t_1,\dots,t_e]\setminus\{0\}$ has been provided for a coupled system in Corollary~\ref{Cor:PiSigmaNested}. Namely, within its proof we determine the aperiodic denominator bound $d$ by applying Theorem~\ref{Thm:GlobalMain} (and thus internally Theorem~\ref{thm:main}) for each \pisiSE-monomial $t_i$. Finally, we merge the different denominator bounds $d_i$ to the global aperiodic denominator bound $d=\lcm(d_1,\dots,t_e)$. In other words, the formula~\eqref{Equ:DenBoundFormula} is reused $e$ times (with possibly different $D$s). This observation gives rise to the following improvement: it suffices to compute for $1\leq i\leq e$ the dispersions $D_i$ (using the formula~\eqref{Equ:FormualForD} for the different \pisiSE-monomials $t_i$), to set $D=\max(D_1,\dots,D_e)$ and to apply only once the formula~\eqref{Equ:DenBoundFormula}. We will illustrate this tactic in an example of Section~\ref{sec:examples}.
\end{remark}

For the sake of clarity we split the proof into two lemmata.

\begin{lemma}\label{lem:disp}
  With the notations of \autoref[Theorem]{thm:main}, it is
  \begin{equation*}
    \disp(\aperiodic(d)) 
    \leq \disp(\sigma^{-\ell}(\aperiodic(m)), \aperiodic(p)) = D.
  \end{equation*}
\end{lemma}

\begin{proof}
  For the ease of notation, we will simply write $\overline{p}$
  instead of $\aperiodic(p)$ and we will do the same with
  $\overline{m} = \aperiodic(m)$ and $\overline{d} =
  \aperiodic(d)$.

  Assume that $\disp(\overline{d}) = \lambda > D$ for some $\lambda
  \in \N$. Then we can find two irreducible aperiodic factors $a,g \in
  \F[t]$ of $\overline{d}$ such that $\sigma^\lambda(a)/g \in \F$. In particular, due to Lemma~\ref{Lemma:PeriodicChar} 
  we can choose $a,g$ with this property such that $\lambda$ is maximal.

  We distinguish two cases. First, assume that $a \mid
  \overline{p}$. We claim that in this case we have $\sigma^\ell(g)
  \nmid \overline{m}$. Otherwise, it was $g \mid
  \sigma^{-\ell}(\overline{m})$ which together with $g \mid
  \sigma^\lambda (a) \mid \sigma^\lambda (\overline{p})$ implied
  $\lambda \in \spread(\sigma^{-\ell} (\overline{m}), \overline{p})$
  which contradicts $D < \lambda$. Moreover, $\sigma^\ell(g)$ cannot
  occur in $\sigma^i(\overline{d})$ for $0 \leq i < \ell$ because else
  $\sigma^\ell(g) \mid \sigma^i (\overline{d})$ and thus $\tilde b =
  \sigma^{\ell - i} (g) \mid \overline{d}$ implied that $a$ and
  $\tilde g$ are factors of $\overline{d}$. Now, since
  $\sigma^{\lambda+\ell-i}(a)/\tilde g =
  \sigma^{\ell-i}(\sigma^\lambda (a)/g) \in \F$, this contradicts the
  maximality of $\lambda$. Thus, $\sigma^\ell (g)$ must occur in the
  denominator of
  \begin{equation}
    \label{eq:proof}
    A_\ell \sigma^\ell (y) 
    + A_{\ell-1} \sigma^{\ell-1} (y)
    \ldots
    + A_1 \sigma (y)
    + A_0 y
    = b
    \in \CV{\F[t]}n
  \end{equation}
  for at least one component: Let $A_\ell^{-1} = \overline{m} U$ for some $U \in
  \Mat{\F[t]}nn$. Then $U A_\ell = \overline{m} \ID[n]$ and
  \begin{multline*}
    \underbrace{U A_\ell}_{= \overline{m}\ID} \sigma^\ell (y) 
    + U A_{\ell-1} \sigma^{\ell-1} (y)
    \ldots
    + U A_1 \sigma (y)
    + U A_0 y
    \\
    = \frac{\overline{m} \sigma^\ell(z)}{\alpha \sigma^\ell(g)} 
    + \frac{
      \sum_{j\neq\ell}  
      \Bigl( \prod_{k \neq j,\ell} \sigma^j(\overline{d}) \Bigr) U A_j \sigma^j(z)
    }{\prod_{j \neq \ell} \sigma^j(\overline{d})}
    = U b
    \in \CV{\F[t]}n
  \end{multline*}
  for some $\alpha \in \CV{\F[t]}n$ such that
  $\sigma^\ell(\overline{d}) = \alpha \sigma^\ell(g)$. The equation is
  equivalent to
  \begin{equation*}
    \Bigl(\prod_{j\neq \ell} \sigma^j(\overline{d}) \Bigr) \overline{m} \sigma^\ell(z) 
    =
    \Bigl(\Bigl(\prod_{j\neq \ell} \sigma^j(\overline{d}) \Bigr) U b 
    - 
    \sum_{j\neq\ell} 
    \Bigl( \prod_{k \neq j,\ell} \sigma^j(\overline{d}) \Bigr) U A_j \sigma^j(z)
    \Bigr)
    \alpha \sigma^\ell(g).
  \end{equation*}
  Note that (every component of the vector on) the \RHS\ is divisible
  by $\sigma^\ell(g)$. For the \LHS, we have
  \begin{equation*}
    \sigma^\ell(g) 
    \Bignmid
    \overline{m} \prod_{j \neq \ell} \sigma^j(\overline{d}).
  \end{equation*}
  Also, we know that $g \nmid z_j$ for at least one $j$. Thus,
  $\sigma^\ell(g)$ does not divide (at least one component of) the
  \LHS. This is a contradiction.

  We now turn our attention to the second case $a \nmid
  \overline{p}$. Here, we consider the related tail regular system
  \begin{math}
    \tilde{A}_{\tilde\ell} \sigma^{\tilde\ell}(y)
    + \ldots + 
    \tilde{A}_0 y = \tilde{b}
  \end{math}
  instead of the original system. Recall that $y$ remains unchanged
  due to \autoref[Remark]{rem:full.reg}. Similar to the first case,
  let $\tilde{A}_0^{-1} = \overline{p} V$, \ie, $V \tilde{A}_0 =
  \overline{p} \ID[n]$ for some $V \in \Mat{\F[t]}nn$. Note that $a
  \nmid \sigma^i(\overline{d})$ for all $i \geq 1$; otherwise,
  $\sigma^{-i}(a)$ was a factor of $\overline{d}$ with
  $\sigma^{\lambda+i}(\sigma^{-i}(a))/b \in \F$ contradicting the
  maximality of $\lambda$. Let now
  \begin{equation*}
    V \tilde A_{\tilde\ell} \sigma^{\tilde\ell}(y)
    + \ldots
    + V \tilde A_1 \sigma(y)
    + \overline{p} \ID[n] y
    = \tilde \xi \in \CV{\F[t]}n.
  \end{equation*}
  We write again $y = \overline{d}^{-1} z$. Then, after multiplying
  with the common denominator $\overline{d} \sigma(\overline{d})
  \cdots \sigma^\ell(\overline{d})$ and rearranging the terms we
  obtain
  \begin{equation*}
    \overline{p} \left( \prod_{k\neq 0} \sigma^k(\overline{d}) \right) z
    = \Bigl(\prod_{j=0}^\ell \sigma^j(\overline{d})\Bigr) \tilde\xi 
    - \sum_{j=1}^\ell 
    \Bigl( \prod_{k\neq j} \sigma^k(\overline{d}) \Bigr) 
    V \tilde A_j \sigma^j(z)
  \end{equation*}
  where every term on the \RHS\ is divisible $a$. However, on the
  \LHS\ $a$ does not divide the scalar factor $\overline{p}
  \prod_{k\neq0} \sigma^k(\overline{d})$ and because of $\gcd(z,
  \overline{d}) = 1$ there is at least one component of $z$ which is
  not divisible by $a$. Thus, $a$ does not divide the left hand side
  which is a contradiction.\qed
\end{proof}

\begin{lemma}\label{lem:bound}
  With the notations of \autoref[Theorem]{thm:main}, we have
  \begin{equation*}
    \aperiodic(d) 
    \Bigmid
    \gcd\Bigl(
    \prod_{j=0}^D \sigma^{-\ell-j}(m),
    \prod_{j=0}^D \sigma^j(p)
    \Bigr).
  \end{equation*}
\end{lemma}

\begin{proof}
  Again, we will simply write $\overline{p}$, $\overline{m}$ and
  $\overline{d}$ instead of $\aperiodic(p)$, $\aperiodic(m)$ and
  $\aperiodic(d)$, respectively.  As in the proof of
  \autoref[Lemma]{lem:disp}, we let $U \in \Mat{\F[t]}nn$ be such that
  $U A_\ell = \overline{m} \ID$. Multiplication by $U$ from the left
  and isolating the highest order term transforms the
  system~\eqref{eq:system} into
  \begin{equation}
    \label{eq:sigma^ell.y}
    \sigma^\ell(y) 
    = \frac1{\overline{m}} U \Bigl(b - \sum_{j=0}^{\ell-1} A_j \sigma^j(y) \Bigr).
  \end{equation}
  Now, we apply $\sigma^{-1}$ to both sides of the equation in order
  to obtain an identity for $\sigma^{\ell-1}(y)$
  \begin{equation*}
    \sigma^{\ell-1}(y) 
    = \frac1{\sigma^{-1} (\overline{m})} 
    \sigma(U) \Bigl(\sigma(b) 
    - \sum_{j=0}^{\ell-1} \sigma(A_j) \sigma^{j-1}(y) \Bigr).
  \end{equation*}
  We can substitute this into~\eqref{eq:sigma^ell.y} in order to
  obtain a representation
  \begin{multline*}  
    \sigma^\ell(y)
    = \frac1{\overline{m}} U \Bigl(b 
    - A_{\ell-1} \frac1{\sigma^{-1}(\overline{m})} \sigma(U) \Bigl(\sigma (b) 
    - \sum_{j=0}^{\ell-1} \sigma (A_j) \sigma^{j-1} (y) \Bigr)
    - \sum_{j=0}^{\ell-2} A_j \sigma^j (y) \Bigr)
    \\
    = \frac1{\overline{m} \, \sigma^{-1}( \overline{m})} \tilde U \Bigl(\tilde b 
    - \sum_{j=-1}^{\ell-2} \tilde A_j \sigma^j (y) \Bigr)
  \end{multline*}
  for $\sigma^\ell (y)$ in terms of $\sigma^{\ell-2} (y), \ldots,
  \sigma^{-1} (y)$ where $\tilde b \in \CV{\F[t]}n$ and $\tilde
  A_{\ell-2},\ldots,\tilde A_{-1}, \tilde U \in \Mat{\F[t]}nn$.

  We can continue this process shifting the terms on the right side
  further with each step. Eventually, after $D$ steps, we will arrive
  at a system of the form
  \begin{equation}
    \label{eq:final}
    \sigma^\ell (y)
    =
    \frac1{\overline{m}\;\sigma^{-1}(\overline{m})\cdots\sigma^{-D} (\overline{m})}
    U' \Bigl(
    b' - \sum_{j=-D}^{\ell-D-1} A'_j \sigma^j (y)
    \Bigr)
  \end{equation}
  where $b' \in \CV{\F[t]}n$ and $A'_{-D},\ldots,A'_{\ell-D-1}, U' \in
  \Mat{\F[t]}nn$.

  Assume now that $y = \overline{d}^{-1} z$ is a solution of~\eqref{eq:system} or thus of~\eqref{eq:final} which is  
  in reduced representation for some
  $\overline{d} \in \F[t]$ and a vector $z \in \CV{\F[t]}n$. Substituting this in
  equation~\eqref{eq:final} yields
  \begin{multline*}
    \frac1{\sigma^\ell(\overline{d})} \sigma^\ell (z)
    = \frac1{\overline{m} \sigma^{-1}(\overline{m}) \ldots \sigma^{-D}(\overline{m})}
    U' \Bigl(
    b' - \sum_{j=-D}^{\ell-D-1} A_j' \frac1{\sigma^j (\overline{d})} \sigma^j (z)
    \Bigr)  
    \\
    = \frac1{\prod_{j=0}^D \sigma^{-j} (\overline{m}) \cdot 
      \prod_{j=-D}^{\ell-D-1} \sigma^j (\overline{d})}
    U' \Bigl(
    \prod_{j=-D}^{\ell-D-1} \sigma^j (\overline{d}) \cdot b' - 
    \sum_{j=-D}^{\ell-D-1} \prod_{k\neq j} \sigma^k (\overline{d}) \cdot A_j' \sigma^j (z)
    \Bigr)  
  \end{multline*}
  or, equivalently after clearing denominators,
  \begin{multline}
    \label{eq:cleared}
    \prod_{j=0}^D \sigma^{-j} (\overline{m}) 
    \cdot \prod_{j=-D}^{\ell-D-1} \sigma^j (\overline{d}) 
    \cdot \sigma^\ell (z)
    \\
    =
    \sigma^\ell (\overline{d}) \cdot
    U' \Bigl(
    \prod_{j=0}^{\ell-1} \sigma^{j-D} (\overline{d}) \cdot b' - 
    \sum_{j=-D}^{\ell-D-1} \prod_{k\neq j} \sigma^k (\overline{d})  \cdot A_j' \sigma^j (z)
    \Bigr).
  \end{multline}
  Let further $q \in \F[t]$ be an irreducible factor of the aperiodic
  part of $\overline{d}$. Then $\sigma^\ell (q)$ divides the right hand side of
  equation~\eqref{eq:cleared}. Looking at the left hand side, we
  see that $\sigma^\ell (q)$ cannot divide $\prod_{j=0}^{\ell-1}
  \sigma^{j-D} (\overline{d})$ since $D = \disp(\overline{d})$ and there is at least one
  entry $z_k$ of $z$ with $1 \leq k \leq n$ such that $q \nmid z_k$
  because $\overline{d}^{-1} z$ is in reduced representation. It follows that
  $\sigma^\ell (q) \mid \prod_{j=0}^D \sigma^{-j} (\overline{m})$, or,
  equivalently,
  \begin{math}
    q \mid \prod_{j=0}^D \sigma^{-\ell-j}(\overline{m}).
  \end{math}
  We can thus cancel $q$ from the equation. Reasoning similarly for
  the other irreducible factors of the aperiodic part of $\overline{d}$ we obtain
  $\overline{d} \mid \prod_{j=0}^D \sigma^{-\ell-j}(\overline{m})$.

  \medskip

  In order to prove $\overline{d} \mid \prod_{j=0}^D \sigma^j(\overline{p})$, we consider
  once more the related tail regular system
  \begin{math}
    \tilde{A}_{\tilde\ell} \sigma^{\tilde\ell}(y)
    + \ldots + 
    \tilde{A}_0 y = \tilde{b}.
  \end{math}
  Recall that by \autoref[Remark]{rem:full.reg} $y$ is both a solution
  of the original and the related. Let $V \tilde{A}_0 = \overline{p} \ID$ for
  some $V \in \Mat{\F[t]}nn$. Multiplying the related system by $V$
  and isolating $y$ yields
  \begin{equation*}
    y = \frac1p V \Bigl( \tilde{b}
    - \sum_{j=1}^{\tilde\ell} \tilde{A}_j \sigma^j(\tilde{y}) \Bigr).
  \end{equation*}
  Now, an analogous computation allows us to shift the orders of the
  terms on the \RHS\ upwards. We obtain an equation
  \begin{equation*}
    y = \frac1{\overline{p} \sigma(\overline{p}) \cdots \sigma^D(\overline{p})} V'
    \Bigl(
    \tilde{b}' - 
    \sum_{j=1}^{\tilde\ell} \tilde{A}_j' \sigma^{D+j}(y)
    \Bigr)
  \end{equation*}
  for suitable matrices $V', \tilde{A}_1', \ldots,
  \tilde{A}_{\tilde{\ell}}' \in \Mat{\F[t]}nn$ and $\tilde{b}' \in
  \CV{\F[t]}n$. Substituting again $y = \overline{d}^{-1} z$ and clearing
  denominators we arrive at an equation similar to~\eqref{eq:cleared}
  and using the same reasoning we can show that $\overline{d} \mid \prod_{j=0}^D
  \sigma^j(\overline{p})$.\qed
\end{proof}

\section{Row and Column Reduction}\label{sec:row-red}

We will show in \autoref[Section]{sec:regularise} below that it is
actually possible to make any system of the form~\eqref{eq:system}
fully regular. One of the key ingredients for this will be row (and
column) reduction which we are going to introduce in this section. The
whole exposition closely follows the one
in~\cite{BeckermannChengLabahn2006}. We will concentrate on row
reduction since column reduction works \latin{mutatis mutandis}
basically the same.

Consider an arbitrary operator matrix $A \in \Mat{\Ore}mn$. When we
speak about the \emph{degree} of $A$, we mean the maximum of the
degrees (in $\sigma$) of all the entries of $A$. Similarly, the degree
of a row of $A$ will be the maximum of the degrees in that row.

Let $\nu$ be the degree of $A$ and let $\nu_1,\ldots,\nu_m$ be the
degrees of the rows of $A$. For simplicity, we first assume that none
of the rows of $A$ is zero. When we multiply $A$ by the matrix $\Delta
= \diag(\sigma^{\nu-\nu_1}, \ldots, \sigma^{\nu-\nu_m})$ from the
left, then for each $i=1,\ldots,m$ we multiply the $i$\th\ row by
$\sigma^{\nu-\nu_i}$. The resulting row will have degree $\nu$. That
is, multiplication by $\Delta$ brings all rows to the same degree. We
will write the product as
\begin{equation*}
  \Delta A = A_\nu \sigma^\nu + \ldots + A_1 \sigma + A_0
\end{equation*}
where $A_0,\ldots,A_\nu \in \Mat{\F(t)}mn$ are rational
matrices. Since none of the rows of $A$ is zero, also none of the rows
of $A_\nu$ is zero. We call $A_\nu$ the \emph{leading row coefficient
  matrix} of $A$ and denote it by $A_\nu = \LRCM(A)$. In general, if
some rows of $A$ are zero, then we simply define the corresponding
rows in $\LRCM(A)$ to be zero, too.

\begin{definition}\label{def:row-red}
  The matrix $A \in \Mat{\Ore}mn$ is \emph{row reduced} (\wrt\
  $\sigma$) if $\LRCM(A)$ has full row rank.
\end{definition}

\begin{remark}\label{ref:row-red->head-reg}
  If $A(y) = b$ is a head reduced system where
  \begin{math}
    A = A_\ell \sigma^\ell + \ldots + A_1 \sigma + A_0
  \end{math}
  for $A_0,\ldots,A_\ell \in \Mat{R}nn$, then $A$ is row-reduced. This
  is obvious since in this case $\LRCM(A) = A_\ell$ and $\det A_\ell
  \neq 0$. Conversely, if $A$ is row-reduced, then $\Delta A$ (with
  $\Delta$ as above) is head regular.
\end{remark}

It can be shown that for any matrix $A \in \Mat{\Ore}mn$ there exists
a unimodular operator matrix $P \in \MatGr{\Ore}m$ such that
\begin{equation*}
  P A =
  \begin{pmatrix}
    \tilde{A} \\ \ZERO
  \end{pmatrix}
\end{equation*}
for some row reduced $\tilde{A} \in \Mat{\Ore}rn$ where $r$ is the
(right row) rank of $A$ over $\Ore$. (For more details, see
\cite[Thm.~2.2]{BeckermannChengLabahn2006} and
\cite[Thm.~A.2]{BeckermannChengLabahn2006}.)

It is a simple exercise to derive an analogous column reduction of
$A$. Moreover, it can easily be shown that it has similar properties.
In particular, we can always compute $Q \in \MatGr{\Ore}n$ such that
the product will be
\begin{equation*}
  A Q =
  \begin{pmatrix}
    \dtilde{A} & \ZERO 
  \end{pmatrix}
\end{equation*}
for some column reduced $\dtilde{A} \in \Mat{\Ore}mr$ where $r$ is the
(left column) rank of $A$.

\medskip

We remark that in fact $r$ in both cases will be the same number since
the left column rank of $A$ equals the right row rank by, \eg,
\cite[Thm.~8.1.1]{cohn}. Therefore, in the following discussion we
will simply refer to it as the \emph{rank} of $A$.

\section{Regularisation}\label{sec:regularise}

In \autoref[Theorem]{thm:main}, we had assumed that we were dealing
with a fully regular system. This section will explain how every
arbitrary system can be transformed into a fully regular one with the
same set of solutions.

\medskip

Represent the system~\eqref{eq:system} by an operator matrix $A \in
\Mat{\Ore}mn$. We first apply column reduction to $A$ which gives a
unimodular operator matrix $Q \in \MatGr{\Ore}n$ such that the
non-zero columns of $A Q$ are column reduced. Next, we apply row
reduction to $A Q$ obtaining $P \in \MatGr{\Ore}m$ such that in total
\begin{equation*}
  P A Q =
  \begin{pmatrix}
    \tilde{A} & \ZERO \\
    \ZERO & \ZERO
  \end{pmatrix}
\end{equation*}
where $\tilde{A} \in \Mat{\Ore}rr$ will now be a row reduced square
matrix and $r$ is the rank of $A$. 

If we define the matrix $\Delta$ as in the previous
\autoref[Section]{sec:row-red}, then the leading coefficient matrix of
$\Delta P A Q$ and that of $P A Q$ will be the same. Moreover, since
$\Delta$ is unimodular over $\OLP$, also $\Delta P$ is unimodular over
$\OLP$. Thus, if we define $\dtilde{A} \in \Mat{\Ore}rr$ by
\begin{equation*}
  \Delta P A Q = 
  \begin{pmatrix}
    \dtilde{A} & \ZERO \\
    \ZERO & \ZERO
  \end{pmatrix},
\end{equation*}
then we can write
\begin{equation*}
  \dtilde{A} 
  = \dtilde{A}_\nu \sigma^\nu + \ldots + \dtilde{A}_1 \sigma + \dtilde{A}_0
\end{equation*}
where $\nu$ is the degree of $\dtilde{A}$ and where
$\dtilde{A}_0,\ldots,\dtilde{A}_\nu \in \Mat{\F(t)}rr$ are rational
matrices. Since $\dtilde{A}$ is still row reduced, we obtain that its
leading row coefficient matrix $\dtilde{A}_\nu$ has full row rank.

Assume now that we started with the system $A(y) = b$. Then $(\Delta P
A Q)(y) = (\Delta P)(b)$ is a related system with the same solutions
as per \autoref[Remark]{rem:ore}. More concretely, let us write
\begin{equation*}
  y =
  \begin{pmatrix}
    y_1 \\ y_2
  \end{pmatrix}
  \qqtext{and}
  (\Delta P)(b) =
  \begin{pmatrix}
    \tilde{b}_1 \\ \tilde{b}_2
  \end{pmatrix}
\end{equation*}
where $y_1$ and $\tilde{b}_1 \in \CV{\F(t)}r$ are vectors of length
$r$, $y_2 \in \CV{\F(t)}{m-r}$ has length $m-r$, and $\tilde{b}_2 \in
\CV{\F(t)}{n-r}$ has length $n-r$. Then $(\Delta P A Q)(y) = (\Delta
P)(b)$ can expressed as
\begin{equation*}
  \dtilde{A}(y_1) = \tilde{b}_1
  \qqtext{and}
  0 = \tilde{b}_2.
\end{equation*}
The requirement that $\tilde{b}_2 = 0$ is a necessary condition for
the system to be solvable. We usually refer to it as a
\emph{compatibility condition}. Moreover, we see that the variables in
$y_2$ can be chosen freely.

If the compatibility condition does not hold, then the system does not
have any solutions and we may abort the computation right
here. Otherwise, $A(y) = b$ is equivalent to a system $\dtilde{A}(y_1)
= \tilde{b}_1$ of (potentially) smaller size. Clearing denominators in
the last system does not change its solvability nor the fact that
$\dtilde{A}$ is row reduced. Thus, we have arrived at an equivalent
head regular system.

\medskip

It remains to explain how we can turn a head regular system into a
fully regular one. Thus, as above we assume now that the first
regularisation step is already done and that the operator matrix $A
\in \Mat{\Ore}nn$ is such that $A(y) = b$ is head regular. That does
in particular imply that $A$ is row reduced and hence that $n$ equals
the rank of $A$ over $\Ore$.

We claim that $n$ is also the rank of $A$ over $\OLP$, \ie, that the
rows of $A$ are linearly independent over $\OLP$. Assume that $v A =
0$ for some $v \in \CV{\OLP}n$. There is a power $\sigma^k$ of
$\sigma$ such that $\sigma^k v \in \CV{\Ore}n$. Since then $(\sigma^k
v) A = 0$, we obtain that $A$ did not have full rank over $\Ore$. The
claim follows by contraposition. Note that also the other direction
obviously holds.

Let $\ell$ be the degree of $A$ and write $A$ as
\begin{equation*}
  A = A_\ell \sigma^\ell + \ldots + A_1 \sigma + A_0
\end{equation*}
where $A_0, \ldots, A_\ell \in \Mat{\F(t)}nn$. We multiply $A$ over
$\OLP$ by $\sigma^{-\ell}$ from the left. This does not change the
rank. The resulting matrix $\sigma^{-\ell} A$ will be in
$\Mat{\OreInv}nn$. Using a similar argument as above, we see that the
rank of $\sigma^{-\ell} A$ over $\OreInv$ is still $n$. We have
\begin{equation*}
  \sigma^{-\ell} A_\ell
  = 
  \sigma^{-\ell}(A_0) \sigma^{-\ell}
  + \ldots
  + \sigma^{-\ell}(A_{\ell-1}) \sigma^{-1}
  + \sigma^{-\ell}(A_\ell),
\end{equation*}
\ie, $\sigma^{-\ell} A$ is similar to $A$ with the coefficients in
reverse order. 

We can now apply row reduction to $\sigma^{-\ell} A$ \wrt\ the Ore
variable\footnote{Note that the commutation rule
  $\sigma^{-1} a = \sigma^{-1}(a)\sigma^{-1}$ follows immediately from
  the rule $\sigma a = \sigma(a) \sigma$.} $\sigma^{-1}$. Just as before we may also
shift all the rows afterwards to bring them to the same degree. Let
the result be
\begin{equation*}
  W \sigma^{-\ell} A 
  = \tilde{A}_0 \sigma^{-\tilde\ell}
  + \ldots
  + \tilde{A}_{\tilde\ell-1} \sigma^{-1}
  + \tilde{A}_{\tilde\ell}
\end{equation*}
where $\tilde\ell$ is the new degree, $W \in \MatGr{\OLP}n$ is a
unimodular operator matrix, the matrices
$\tilde{A}_0,\ldots,\tilde{A}_{\tilde\ell} \in \Mat{\F(t)}nn$ are
rational, and where the non-zero rows of $\tilde{A}_0$ are
independent. However, since the rank of $\sigma^{-\ell} A$ is $n$
(over $\OreInv$), we obtain that $\tilde{A}_0$ does in fact not
possess any zero-rows at all. Thus, $\tilde{A}_0$ has full rank. 

Multiplication by $\sigma^{\tilde\ell}$ from the left, brings
everything back into $\Mat{\Ore}nn$; \ie, we have
\begin{equation*}
  \sigma^{\tilde\ell} W \sigma^{-\ell} A 
  = \sigma^{\tilde\ell}(\tilde{A}_{\tilde\ell}) \sigma^{\tilde\ell}
  + \ldots 
  + \sigma^{\tilde\ell}(\tilde{A}_1) \sigma
  + \sigma^{\tilde\ell}(\tilde{A}_0)
\end{equation*}
where $\sigma^{\tilde\ell}(\tilde{A}_0)$ still has full rank and where
the transformation matrix $\sigma^{\tilde\ell} W \sigma^{-\ell}$ is
unimodular over $\OLP$. In other words, we have found a related tail
regular system. Since we started with a head regular system, we even
found that it is fully reduced.

We can summarise the results of this section in the following way. An
overview of the process is shown in \autoref[Figure]{fig:regularise}.

\begin{theorem}\label{thm:regularise}
  Any system of the form~\eqref{eq:system} can be transformed into an
  related fully regular system. Along the way we acquire some
  compatibility conditions indicating where the system may be
  solvable.
\end{theorem}

\begin{figure}
  \centering
  \begin{tikzpicture}
    \matrix (m) [
    matrix of math nodes, 
    row sep=2ex,
    column sep=12em,
    ] {
      A \in \Mat{\Ore}mn \text{ arbitrary} 
      \\[6ex]
      (\Delta P) A Q =
      \begin{pmatrix}
        \dtilde{A} & \ZERO \\
        \ZERO & \ZERO
      \end{pmatrix}
      \makebox[0cm][l]{\text{ with $\dtilde{A} \in \Mat{\Ore}rr$ head regular}}
      \\[6ex]
      \dtilde{A} \text{ head regular}
      & \sigma^{\tilde\ell} W \sigma^{-\ell} \dtilde{A}
      \text{ tail regular}\\
    };
    \draw[->] (m-1-1) -- node[right] {\small row/column reduction w.\,r.\,t.\ $\sigma$} (m-2-1);
    \draw[->] (m-3-1) -- node[below] {\small row reduction w.\,r.\,t.\ $\sigma^{-1}$} 
    (m-3-2);
    \draw[double, double distance=2pt] (m-2-1) -- 
    node[right] {\small assuming the compatibility conditions hold}
    (m-3-1);
  \end{tikzpicture}
  \caption{Outline of the regularisation.}
  \label{fig:regularise}
\end{figure}
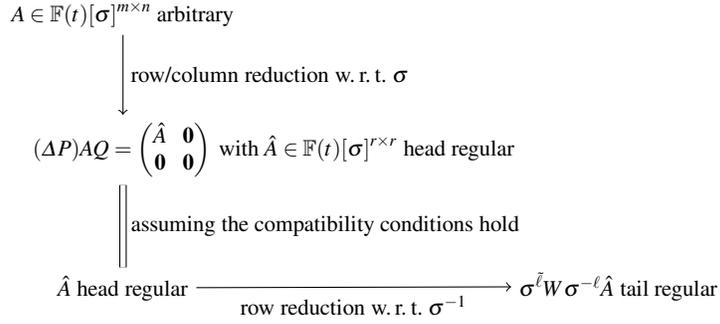

We would like to once more compare our approach to the one taken in
\cite{AbramovBarkatou2014}. They show how to turn a system into
strongly row-reduced form (their version of fully regular as explained
after \autoref[Definition]{def:full.reg} in the proof of their
\cite[Prop.~5]{AbramovBarkatou2014}. Although they start out with
an input of full rank, this is not a severe restriction as the same
preprocessing step (from $A$ to $\hat{A}$) which we used could be
applied in their case, too. Just like our approach, their method
requires two applications of row reduction. They do, however, obtain
full regularity in the opposite order: The first reduction makes the
system tail regular while the second reduction works on the leading
matrix. In our case, the first row reduction (removes unnecessary
equations) and makes the system head regular while the second one
works on the tail. The other big difference is that our second
reduction is \wrt\ $\sigma^{-1}$ while \cite{AbramovBarkatou2014}
rewrite the system in terms of the difference operator $\Delta =
\sigma - 1$ and then reduce \wrt\ $\Delta$. As mentioned after
\autoref[Definition]{def:full.reg}, we cannot with certainty tell yet
which of the two approaches is preferable. That will be a topic for
future research.

\section{Examples}\label{sec:examples}

As a first example, we consider the system
\begin{multline*}
  \begin{pmatrix}
    -2t^2 - t + 1 & 0 \\
    -2t^5 - 9t^4 - 15t^3 - 8t + 3t + 3 
    & -t^7 - 2t^6 - 4t^5 - 6t^4 - 7t^3 - 8t^2 - 4t
  \end{pmatrix}
  y(t+1)
  \\
  +
  \begin{pmatrix}
    t^4 - t^3 + 2t^2 & t^4 - t^3 + 2t^2 \\
    0 & t^7 + 3t^6 + 4t^5 + 5t^4 + 9t^3 + 6t^2 
  \end{pmatrix}
  y(t)
  \\
  =
  \begin{pmatrix}
    0 \\
    2t^5 + 3t^4 + t^3 + 8t^2 + 4t
  \end{pmatrix}.
\end{multline*}
Here, we have $\F = \Q$ and we are in the \SigmaExt\ case with
$\sigma(t) = t + 1$. We can easily see that the leading and trailing
matrices are both regular. Inverting them and computing common
denominators, we arrive at
\begin{equation*}
  m = (2t-1) t (t^2 + t + 2) (t^2 - t + 2) (t+1)^2
\end{equation*}
and
\begin{equation*}
  p = t^2 (t+1) (t^2 - t+ 2) (t^2 + 3t + 3).
\end{equation*}
We have $\spread(\sigma^{-1} m, p) = \{0\}$ and thus the dispersion is
$0$. We obtain the denominator bound
\begin{equation*}
  \gcd(\sigma^{-1} m,p) = t^2 (t^2 - t + 2).
\end{equation*}
This does fit well with the actual solutions for which a $\Q$-basis is
given by
\begin{equation*}
  \frac1{t^2 (t^2 - t + 2)}
  \begin{pmatrix}
    -t (t^2 - t + 2) \\ 
    t^3 - t^2 + 1
  \end{pmatrix}
  \quad\text{and}\quad
  \frac1{t^2 (t^2 - t + 2)}
  \begin{pmatrix}
    -t^3 (t^2 - t + 2) \\ 
    t^5 - t^4 - 3t^2 + 1
  \end{pmatrix}.
\end{equation*}
(We can easily check that those are solutions; and according to
\cite[Thm.~6]{AbramovBarkatou2014} the dimension of the solution space
is $2$.)

\bigskip

For the second example, we consider a $(2,3)$-multibasic rational
difference field over $\F = \Q$; \ie, we consider $\Q(t_1,t_2)$ with
$\sigma(t_1) = 2 t_1$ and $\sigma(t_2) = 3 t_2$. The system in this
example is
\begin{multline*}
  \overbrace{\begin{pmatrix}
      (11 t_1 t_2 - 1) (36 t_1 t_2 - 1) 
      &
      -(11 t_1 t_2 - 1) (36 t_1 t_2 - 1)
      \\
      (4 t_1 - 9 t_2) (2 t_1 - 3 t_2)  
      &
      (4 t_1 - 9 t_2) (2 t_1 - 3 t_2)  
    \end{pmatrix}}^{=: A_2(t_1,t_2)}
  \begin{pmatrix}
    y_1(4 t_1, 9 t_2) \\ 
    y_2(4 t_1, 9 t_2) 
  \end{pmatrix}
  \\
  \;+\;
  \begin{pmatrix}
    -(6 t_1 t_2 - 1) (143 t_1 t_2 - 3)  
    &
    (6 t_1 t_2 - 1) (143 t_1 t_2 - 3)
    \\
    -6 (2 t_1 - 3 t_2) (t_1 - 2 t_2) 
    & 
    -6 (2 t_1 - 3 t_2) (t_1 - 2 t_2) 
  \end{pmatrix}
  \begin{pmatrix}
    y_1(2 t_1, 3 t_2) \\ 
    y_2(2 t_1, 3 t_2) 
  \end{pmatrix}
  \\
  \;+\;
  \underbrace{\begin{pmatrix}
      2 (t_1 t_2 - 1) (66 t_1 t_2 - 1) 
      & 
      -2 (t_1 t_2 - 1) (66 t_1 t_2 - 1)
      \\                                                             
      (4 t_1 - 9 t_2) (t_1 - t_2)  
      &
      (4 t_1 - 9 t_2) (t_1 - t_2)   
    \end{pmatrix}}_{=: A_0(t_1,t_2)}
  \begin{pmatrix}
    y_1(t_1, t_2) \\ 
    y_2(t_1, t_2) 
  \end{pmatrix}
  =
  \begin{pmatrix}
    0 \\
    0
  \end{pmatrix}.
\end{multline*}
This is a $2$-by-$2$ system of order $2$ over
$\mathbb{Q}(t_1,t_2)[\sigma]$. Both the $\sigma$-leading matrix $A_2$
and the $\sigma$-trailing matrix $A_0$ are invertible, which means
that the system is both head and tail regular; hence it is fully
regular. The denominator of $A_2^{-1}$ is
\begin{equation*}
  m = 2 (11 t_1 t_2 - 1) (36 t_1 t_2 - 1) (4 t_1 - 9 t_2) (2 t_1 - 3 t_2)
\end{equation*}
and the denominator of $A_0^{-1}$ is
\begin{equation*}
  p = 4 (t_1 t_2 - 1) (66 t_1 t_2 - 1) (4 t_1 - 9 t_2) (t_1 - t_2).
\end{equation*}
We have $\aperiodic(m)= m$ and $\aperiodic(p) = p$. Following the strategy/algorithm proposed in Remark~\ref{Remark:ImprovedMultivariateDen} we compute the
dispersions \wrt\ $t_1$ and $t_2$ (which turn out to be the same in
this example); obtaining
\begin{equation*}
  D = \disp_{t_1,t_2}(\sigma^{-2}(\aperiodic(m)), \aperiodic(p)) = 0.
\end{equation*}
By \autoref[Corollary]{Cor:PiSigmaNested} it follows that the denominator bound
for this system is
\begin{equation*}
  d 
  = \gcd(\sigma^{-2}(\aperiodic(m)), \aperiodic(d))
  = (t_1 t_2 - 1) (t_1 - t_2).
\end{equation*}
This fits perfectly with the actual $\Q$-basis of the solution space
which is given by
\begin{multline*}
  \frac1{2 (t_1 t_2 - 1) (t_1 - t_2)}
  \begin{pmatrix}
    (t_2 + 1) (t_1 - 1) \\
    (t_2 - 1) (t_1 + 1)
  \end{pmatrix},
  \quad
  \frac1{2 (t_1 - t_2)}
  \begin{pmatrix}
    t_1^2 - t_1 t_2 + 1 \\
    - t_1^2 + t_1 t_2 + 1
  \end{pmatrix},
  \\
  \frac1{4 (t_1 - t_2)}
  \begin{pmatrix}
    2 t_1^2 - 2 t_1 t_2 + 4 t_1 - 3 t_2 \\
    -2 t_1^2 + 2 t_1 t_2 + 4 t_1 - 3 t_2
  \end{pmatrix},
  \\\text{and}\quad
  \frac1{4 (t_1 t_2 - 1) (t_1 - t_2)}
  \begin{pmatrix}
    4 t_1^2 t_2 - 3 t_1 t_2^2 - 2 t_1 + t_2 \\
    4 t_1^2 t_2 - 3 t_1 t_2^2 - 6 t_1 + 5 t_2
  \end{pmatrix}.
\end{multline*}
(It is easy to check that these are solutions; and they are a basis of
the solutions since the dimension of the solution space is $4$
according to \cite{AbramovBarkatou2014}.)

\section{Conclusion}\label{Sec:Conclusion}

Given a \pisiSE-extension $(\F(t),\sigma)$ of $(\F,\sigma)$ and a coupled system of the form~\eqref{eq:system} whose coefficients are from $\F(t)$, we presented algorithms that compute an aperiodic denominator bound $d\in\F[t]$ 
for the solutions under the assumption that the dispersion can be computed in $\F[t]$ (see Theorem~\ref{Thm:GlobalMain}). If $t$ represents a sum, i.e., it has the shift behaviour $\sigma(t)=t+\beta$ for some $\beta\in\F$, this is the complete denominator bound. If $t$ represents a product, i.e., it has the shift behaviour $\sigma(t)=\alpha\,t$ for some $\alpha\in\F^*$, then $t^m\,d$ will be a complete denominator bound for a sufficiently large $m$. It is so far an open problem to determine this $m$ in the $\Pi$-monomial case by an algorithm; so far a solution is only given for the $q$-case with $\sigma(t)=q\,t$ in~\cite{Middeke2017b}. In the general case, one can still guess $m\in\N$, i.e., one can choose a possibly large enough $m$ ($m=0$ if $t$ is a $\Sigma$-monomial) and continue. Namely, plugging $y=\frac{y'}{t^m\,d}$ with the unknown numerator $y'\in\F[t]^n$ into the system~\eqref{eq:system} yields a new system in $y'$ where one has to search for all polynomial solutions $y'\in\F[t]^n$. It is still an open problem to determine a degree bound $b\in\N$ that bounds the degrees of all entries of all solutions $y'$; for the rational case $\sigma(t)=t+1$ see~\cite{AB98} and for the $q$-case $\sigma(t)=q\,t$ see~\cite{Middeke2017b}. In the general case, one can guess a degree bound $b$, i.e., one can choose a possibly large enough $b\in\N$ and continues to find all solutions $y'$ whose degrees of the components are at most $b$. 
This means that one has to determine the coefficients up to degree $b$ in the difference field $(\F,\sigma)$.\\
If $\F=\const(\F,\sigma)$, this task can be accomplished by reducing the problem to a linear system and solving it. Otherwise, suppose that $\F$ itself is a \pisiSE-field over a constant field $\K$. Note that in this case we can compute $d$ (see Lemma~\ref{Cor:DispersionInExtension}), i.e., we only supplemented a tuple $(m,b)$ of nonnegative integers to reach this point.
Now one can use degree reduction strategies as worked out in~\cite{Karr81,Bron:00,Schneider:05a} to determine the coefficients of the polynomial solutions by solving several coupled systems in the smaller field $\F$. In other words, we can apply our strategy again to solve these systems in $\F=\F'(\tau)$ where $\tau$ is again a \pisiSE-monomial: compute the aperiodic denominator bound $d'\in\F'[\tau]$, guess an integer $m'\geq0$ ($m'=0$ if $\tau$ is a $\Pi$-monomial) for a complete denominator bound $\tau^{m'}\,d'$, guess a degree bound $b'\geq0$ and determine the coefficients of the polynomial solutions by solving coupled systems in the smaller field $\F'$. Eventually, we end up at the constant field and solve the problem there by linear algebra. 

Summarising, we obtain a method that enables one to search for all solutions of a coupled system in a \pisiSE-field where one has to adjust certain nonnegative integer tuples $(m,b)$ to guide our machinery. Restricting to scalar equations with coefficients from a \pisiSE-field, the bounds of the period denominator part and the degree bounds has been determined only some years ago~\cite{ABPS:17}. Till then we used the above strategy also for scalar equations~\cite{Schneider:05a} and could derive the solutions in concrete problems in a rather convincing way. It is thus expected that this approach will be also rather helpful for future calculations.


\providecommand{\bysame}{\leavevmode\hbox to3em{\hrulefill}\thinspace}
\providecommand{\MR}{\relax\ifhmode\unskip\space\fi MR }
\providecommand{\MRhref}[2]{%
  \href{http://www.ams.org/mathscinet-getitem?mr=#1}{#2}
}
\providecommand{\href}[2]{#2}

\end{document}